\documentclass[12pt,english]{article}
\usepackage{mathpazo}
\usepackage[T1]{fontenc}
\usepackage[latin9]{inputenc}
\usepackage{geometry}
\usepackage{amsmath}
\usepackage{amsthm}
\usepackage{amssymb}
\usepackage{stmaryrd}
\usepackage{graphicx}

\makeatletter
\theoremstyle{plain}
\newtheorem{thm}{\protect\theoremname}
\theoremstyle{definition}
\newtheorem{defn}{\protect\definitionname}
\theoremstyle{plain}
\newtheorem{lem}{\protect\lemmaname}
\theoremstyle{remark}
\newtheorem*{acknowledgement*}{\protect\acknowledgementname}

\usepackage{mathrsfs}
\usepackage{amsmath}
\usepackage{mathtools}
\usepackage{braket}
\usepackage{upgreek}
\DeclareMathOperator{\Tr}{Tr}
\DeclareMathOperator{\I}{I}

\usepackage{enumitem}
\setenumerate{label=(\roman*)}

\newcommand{\C}{\mathbb{C}}
\newcommand{\R}{\mathbb{R}}

\renewcommand{\sigma}{\upsigma}

\makeatother

\usepackage{babel}
\providecommand{\acknowledgementname}{Acknowledgement}
\providecommand{\definitionname}{Definition}
\providecommand{\lemmaname}{Lemma}
\providecommand{\theoremname}{Theorem}

\begin{document}
\title{Gleason-type Theorems from\\
Cauchy's Functional Equation}
\author{Victoria J Wright and Stefan Weigert\\
 Department of Mathematics, University of York\\
 York YO10 5DD, United Kingdom\\
 \texttt{\small{}vw550@york.ac.uk, stefan.weigert@york.ac.uk}}
\date{May 2019}
\maketitle
\begin{abstract}
Gleason-type theorems derive the density operator and the Born rule
formalism of quantum theory from the measurement postulate, by considering
additive functions which assign probabilities to measurement outcomes.
Additivity is also the defining property of solutions to Cauchy's
functional equation. This observation suggests an alternative proof
of the strongest known Gleason-type theorem, based on techniques used
to solve functional equations.
\end{abstract}
\global\long\def\kb#1#2{|#1\rangle\langle#2|}%

\global\long\def\bk#1#2{\langle#1|#2\rangle}%

\global\long\def\braket#1#2{\langle#1|#2\rangle}%

\global\long\def\ket#1{|#1\rangle}%

\global\long\def\bra#1{\langle#1|}%

\global\long\def\cd{\mathbb{C}^{d}}%

\section{Introduction}

Gleason's theorem \cite{Gleason1957} is a fundamental result in the
foundations of quantum theory simplifying the axiomatic structure
upon which the theory is based. The theorem shows that quantum states
must correspond to density operators if they are to consistently assign
probabilities to the outcomes of projective measurements in Hilbert
spaces of dimension three or larger.\footnote{By a c\emph{onsistent} assignment of probabilities we mean one in
which the probabilities for all outcomes of a given measurement sum
to one.}

More explicitly, let $\mathcal{P\left(H\right)}$ be the lattice of
self-adjoint projections onto closed subspaces of a separable Hilbert
space $\mathcal{H}$ of dimension at least three. Consider functions
$f:\mathcal{P\left(H\right)}\rightarrow\left[0,1\right]$, that are
\emph{finitely additive }for projections $P_{1}$ and $P_{2}$ onto
\emph{orthogonal }subspaces of $\mathcal{H}$, i.e. 
\begin{equation}
f\left(P_{1}\right)+f\left(P_{2}\right)=f\left(P_{1}+P_{2}\right)\,.\label{eq:additive projections}
\end{equation}
Gleason's result shows that the solutions of Eq. (\ref{eq:additive projections})
(in finite dimensional Hilbert spaces) \footnote{Gleason proved this result in all separable Hilbert spaces if the
condition of finite additivity is replaced with $\sigma$-additivity.
These conditions are equivalent in finite dimensional Hilbert spaces.
Later Christensen \cite{christensen} showed that the weaker condition
of finite additivity was also sufficient for the result to hold in
infinite dimensions.} necessarily admit an expression 
\begin{equation}
f\left(\cdot\right)=\Tr\left(\rho\cdot\right),\label{eq:Born}
\end{equation}
for some positive-semidefinite self-adjoint operator $\rho$ on $\mathcal{H}$.

The result does not hold, however, in Hilbert spaces of dimension
two since the constraints (\ref{eq:additive projections}) degenerate
in this case: the projections lack the ``intertwining'' property
\cite{Gleason1957} present in higher dimensions. In 2003, Busch \cite{Busch2003}
and then Caves et al. \cite{Caves2004} extended Gleason's theorem
to dimension two by considering \emph{generalised }quantum measurements
described by positive operator-valued measures, or POMs. In analogy
with Gleason's original requirement, a state is now defined as an
additive probability assignment not only on projections but on a larger
set of operators, the space $\mathcal{E}\left(\mathcal{H}\right)$
of \emph{effects}\footnote{An effect $E$ on $\mathcal{H}$ is a self-adjoint operator satisfying
$0\leq\bk{\psi}{E\psi}\leq\bk{\psi}{\psi}$ for all vectors $\ket{\psi}\in\mathcal{H}$.} defined on a separable Hilbert space. Then, in finite dimensional
Hilbert spaces, any function $f:\mathcal{E\left(H\right)}\rightarrow\left[0,1\right]$
satisfying finite additivity,
\begin{equation}
f\left(E_{1}\right)+f\left(E_{2}\right)=f\left(E_{1}+E_{2}\right)\,,\label{eq:additive effects}
\end{equation}
for effects $E_{1},E_{2}\in\mathcal{E\left(H\right)}$ such that 
\begin{equation}
E_{1}+E_{2}\in\mathcal{E\left(H\right)}\,,\label{eq:coexist}
\end{equation}
is found to necessarily admit an expression of the form given in Eq.
(\ref{eq:Born}).\footnote{This result does \emph{not} imply Gleason's result since in dimensions
greater than two the requirement (\ref{eq:additive effects}) is stronger
than the requirement (\ref{eq:additive projections}).} The effects $E_{1}$ and $E_{2}$ are said to \emph{coexist} since
the condition in Eq. (\ref{eq:coexist}) implies that they occur in
the range of a \emph{single} POM. More recently, it has been shown
that this \emph{Gleason-type theorem}\footnote{It is important to clearly distinguish Gleason\emph{-type} theorems
from Gleason's original theorem.} also follows from weaker assumptions: it is sufficient to require
Eq. (\ref{eq:additive effects}) hold only for effects $E_{1}$ and
$E_{2}$ that coexist in \emph{projective-simulable }measurements
obtained by mixing projective measurements \cite{Wright2018}.

Finitely additive functions were first given serious consideration
in 1821 when Cauchy \cite{cauchy1821cours} attempted to find all
solutions of the equation 
\begin{equation}
f\left(x\right)+f\left(y\right)=f\left(x+y\right)\,,\label{eq:cfe}
\end{equation}
for real variables $x,\,y\in\mathbb{R}$. In addition to the obvious
linear solutions, non-linear solutions to \emph{Cauchy's functional
equation} are known to exist \cite{hamel1905basis}. However, the
non-linear functions $f$ satisfying Eq. (\ref{eq:cfe}) cannot be
Lebesgue measurable \cite{banach1920equation}, continuous at a single
point \cite{Darboux1875continuityatapoint} or bounded on any set
of positive measure \cite{kestelman1947functional}. Similar results
also hold for Cauchy's functional equation with arguments more general
than real numbers, reviewed in \cite{aczel1966lectures}, for example.

Recalling that the Hermitian operators on $\C^{d}$ form a real vector
space, it becomes clear that the Gleason-type theorems described above
can be viewed as results about the solutions of Cauchy's functional
equation for \emph{vector-valued arguments}: additive functions on
subsets of a real vector space, subject to some additional constraints,
are necessarily linear. Taking advantage of this connection, we use
results regarding Cauchy's functional equation to present an alternative
proof of known Gleason-type theorems.

In Sec. 2, we spell out four conditions that single out \emph{linear
}solutions to Cauchy's functional equation defined on a finite interval
of the real line. The main result of this paper---an alternative
method to derive Busch's Gleason-type theorem---is presented in Sec.
3. We conclude with a summary and a discussion of the results in Sec.~\ref{sec:Summary-and-discussion}.

\section{Cauchy's functional equation on a finite interval\label{sec: Cauchy finite interval}}

In 1821 Cauchy \cite{cauchy1821cours} showed that a \emph{continuous}
function over the real numbers satisfying Eq. (\ref{eq:cfe}) is necessarily
linear. It is important to note, however, that relaxing the continuity
restriction does allow for non-linear solutions \cite{hamel1905basis},
as pathological as they may be.\footnote{The existence of non-linear solutions depends on the existence of
Hamel bases and, thus, on the axiom of choice.} Other conditions known to ensure linearity of a finitely additive
function include Lebesgue measurability \cite{banach1920equation},
positivity on small numbers \cite{Darboux1880smallpositive} or continuity
at a single point \cite{Darboux1875continuityatapoint}. We begin
by proving a related result, in which the domain of the function is
restricted to an interval, as opposed to the entire real line.
\begin{thm}
\label{Thm bounded Cauchy}Let $a>0$ and $f:\left[0,a\right]\rightarrow\mathbb{R}$
be a function that satisfies 
\begin{equation}
f\left(x\right)+f\left(y\right)=f\left(x+y\right)\,,\label{eq:additivethm}
\end{equation}
for all $x,y\in\left[0,a\right]$ such that $\left(x+y\right)\in\left[0,a\right]$.
The function $f$ is necessarily linear, i.e. 
\begin{equation}
f\left(x\right)=\frac{f\left(a\right)}{a}x\,,\label{eq: f is linear}
\end{equation}
if it satisfies any one of the following four conditions:
\begin{enumerate}
\item \label{enu: Cauchy 1} $f\left(x\right)\leq b$ for some $b\geq0$
and all $x\in\left[0,a\right]$;
\item \label{enu:Cauchy 2} $f\left(x\right)\geq c$ for some $c\leq0$
and all $x\in\left[0,a\right]$;
\item \label{enu:Cauchy 3} $f$ is continuous at zero;
\item \label{enu:Cauchy 4} $f$ is Lebesgue-measurable.
\end{enumerate}
\end{thm}
Theorem \ref{Thm bounded Cauchy} says that non-linear solutions of
Eq. (\ref{eq:additivethm}) cannot be bounded from below or above,
continuous at zero or Lebesgue measurable. We will now prove the linearity
of $f$ for Case \ref{enu: Cauchy 1}. The proofs for the remaining
cases are given in Appendix~\ref{sec:Appendix}.
\begin{proof}
We will extend $f$ to a finitely additive function on the entire
real line. For any real number $x\in\left[0,a\right]$, Eq. (\ref{eq:additivethm})
implies that 
\begin{equation}
f\left(x\right)=f\left(\frac{n}{n}x\right)=nf\left(\frac{x}{n}\right)\,,\label{eq:rational1}
\end{equation}
where $n$ is a positive integer. If we choose an integer $m\in\mathbb{N}$
with $m/n\in\left[0,a\right]$, then we have 
\begin{equation}
f\left(\frac{m}{n}x\right)=mf\left(\frac{x}{n}\right)=\frac{m}{n}f\left(x\right)\,.\label{eq:rational2}
\end{equation}

In a first step, we extend the function $f$ to all \emph{non-negative}
real numbers by defining 
\begin{equation}
f_{+}(x)=nf\left(\frac{x}{n}\right),
\end{equation}
for real numbers $x>a$ and integers $n>x/a$. This extension is well-defined
since for any two sufficiently large integers, i.e. $m$ and $n$
with $m,n>x/a$, we have 
\begin{equation}
f\left(\frac{x}{mn}\right)=\frac{1}{m}f\left(\frac{x}{n}\right)=\frac{1}{n}f\left(\frac{x}{m}\right),
\end{equation}
according to Eq. (\ref{eq:rational1}), resulting in the identity
\begin{equation}
mf\left(\frac{x}{m}\right)=nf\left(\frac{x}{n}\right).
\end{equation}
Finite additivity on the positive half-line also holds since for any
two non-negative numbers $x,y\geq0$, we find 
\begin{equation}
\begin{aligned}f_{+}(x)+f_{+}(y) & =nf\left(\frac{x}{n}\right)+nf\left(\frac{y}{n}\right)\\
 & =nf\left(\frac{x+y}{n}\right)\\
 & =f_{+}(x+y)\,,
\end{aligned}
\end{equation}
for sufficiently large $n\in\mathbb{N}$ which ensures that $(x+y)/n\in[0,a]$.

In a second step, we extend the function $f_{+}$ to the \emph{entire}
real line by defining
\begin{equation}
f_{\mathbb{R}}(x)=\begin{cases}
f_{+}(x) & x\geq0\,,\\
-f_{+}(-x) & x<0\,.
\end{cases}\label{eq: def of f_R}
\end{equation}
 To show that the function $f_{\mathbb{R}}$ is finitely additive
on all of $\mathbb{R}$, three cases must be considered.

If both $x<0$ and $y<0$, we have 
\begin{equation}
\begin{aligned}f_{\mathbb{R}}(x)+f_{\mathbb{R}}(y) & =-f_{+}(-x)-f_{+}(-y)\\
 & =-f_{+}(-x-y)\\
 & =f_{\mathbb{R}}(x+y)\,,
\end{aligned}
\end{equation}
using that $f_{+}(-x)+f_{+}(-y)=f_{+}(-x-y)$ holds for non-negative
real numbers $-x$ and $-y$.

If $x\geq0$, $y<0$ and $x+y<0$, we have 
\begin{equation}
\begin{aligned}f_{\mathbb{R}}(x)+f_{\mathbb{R}}(y) & =f_{+}(x)-f_{+}(-y-x+x)\\
 & =f_{+}(x)-f_{+}(-y-x)-f_{+}(x)\\
 & =f_{\mathbb{R}}(x+y)\,.
\end{aligned}
\end{equation}

If $x\geq0$, $y<0$ and $x+y\geq0$, we have 
\begin{equation}
\begin{aligned}f_{\mathbb{R}}(x)+f_{\mathbb{R}}(y) & =f_{+}(x+y-y)-f_{+}(-y)\\
 & =f_{+}(x+y)+f_{+}(-y)-f_{+}(-y)\,\\
 & =f_{\mathbb{R}}(x+y)\,.
\end{aligned}
\end{equation}

This property completes the proposed extension of the function $f$
to a finitely additive function $f_{\mathbb{R}}$ on the real line
that is bounded above on the interval $[0,a]$. Ostrowski \cite{ostrowski1929mathematische}
and Kestelman \cite{kestelman1947functional} showed that finitely
additive functions on the real line that are bounded above on a \emph{set
of positive measure} are necessarily linear. Therefore, the extended
function $f_{\mathbb{R}}$ is linear, and its restriction back to
the interval $[0,a]$ is given by $f(x)=f(a)x/a$.
\end{proof}

\section{When Cauchy meets Gleason: additive functions on effect spaces}

The first Gleason-type theorem, published in 2003, assumes additivity
of the frame function not only on projections that occur in the same
\emph{projection-valued measure} (PVM) but on the larger set of effects
that coexist in the same POM.
\begin{thm}[Busch \cite{Busch2003}]
 \label{thm: Busch}Let $\mathcal{E}_{d}$ be the space of effects
on $\C^{d}$ and $\I_{d}$ be the identity operator on $\C^{d}$.
Any function $f:\mathcal{E}_{d}\rightarrow\left[0,1\right]$ satisfying
\begin{equation}
f\left(\I_{d}\right)=1\,,\label{eq:identity}
\end{equation}
and
\begin{equation}
f\left(E_{1}\right)+f\left(E_{2}\right)=f\left(E_{1}+E_{2}\right)\,,\label{eq:additivity}
\end{equation}
for all $E_{1},E_{2}\in\mathcal{E}_{d}$ such that $\left(E_{1}+E_{2}\right)\in\mathcal{E}_{d}$,
admits an expression 
\begin{equation}
f\left(E\right)=\Tr\left(E\rho\right),\label{eq:busch density-2-2-1}
\end{equation}
for some density operator $\rho$, and all effects $E\in\mathcal{E}_{d}$.
\end{thm}
Theorem \ref{thm: Busch} rephrases the (finite-dimensional case of
the) theorem proved by Busch \cite{Busch2003} and the theorem due
to Caves et al. \cite{Caves2004}. Busch uses the positivity of the
frame function $f$ to directly establish its homogeneity whereas
Caves et al. derive homogeneity by showing that the frame function
$f$ must be continuous at the zero operator. These arguments seem
to run in parallel with Cases \ref{enu:Cauchy 2} and \ref{enu:Cauchy 3}
of Theorem \ref{Thm bounded Cauchy} presented in the previous section.
In Sec. \ref{subsec:Alternative-proof}, we will give an alternative
proof of Theorem \ref{thm: Busch} which can be based on any of the
four cases of Theorem \ref{Thm bounded Cauchy}.

\subsection{Preliminaries}

To begin, let us introduce a number of useful concepts and establish
a suitable notation. Throughout this section we will make use of the
fact that the Hermitian operators on $\mathbb{C}^{d}$ constitute
a real vector space of dimension $d^{2}$, which we will denote by
$\mathbb{H}_{d}$. We may therefore employ the standard inner product
$\left\langle A,B\right\rangle =\Tr\left(AB\right)$, for Hermitian
operators $A$ and $B$, in our reasoning as well as the norm $\left\Vert \cdot\right\Vert $
which it induces.

A discrete POM on $\C^{d}$ is described by its range, i.e. by a sequence
of effects $\left\llbracket E_{1},E_{2},\dots\right\rrbracket $ that
sum to the identity operator on $\C^{d}$. A \emph{minimal informationally-com-plete}
(MIC) POM $\mathcal{M}$ on $\C^{d}$ consists of exactly $d^{2}$
linearly independent effects, ${\cal M}=\left\llbracket M_{1},\ldots,M_{d^{2}}\right\rrbracket $.
Hence, MIC-POMs constitute bases of the vector space of Hermitian
operators, and it is known that they exist in all finite dimensions
\cite{caves2002unknown}.

Positive linear combinations of effects will play an important role
below, giving rise to the following definition.
\begin{defn}
The \emph{positive cone} of a set of Hermitian operators $S$ on $\C^{d}$
is the set of non-negative linear combinations of the elements of
$S$, i.e. the set
\begin{equation}
\mathcal{C}\left(S\right)=\left\{ H=\sum_{j=1}^{d^{2}}a_{j}H_{j}\,,\,a_{j}\geq0\text{ for any }H_{j}\in S\right\} .
\end{equation}
Note that the expression of an element of $\mathcal{C}\left(S\right)$
as a linear combination of elements of $S$ requires at most $d^{2}$
terms as a consequence of Caratheodory's theorem.
\end{defn}
Next, we introduce so-called ``augmented'' bases of the space $\mathbb{H}_{d}$
which are built around sets of $d$ projections $\left\{ \ket{e_{1}}\bra{e_{1}},\ldots,\ket{e_{d}}\bra{e_{d}}\right\} $
where the vectors $\left\{ \ket{e_{1}},\ldots\ket{e_{d}}\right\} $
form an orthonormal basis of $\C^{d}$.
\begin{defn}
\label{def:An-augmented-basis}An \emph{augmented basis }of the Hermitian
operators on $\cd$ is a set of $d^{2}$ linearly independent rank-one
effects $\mathcal{B}=\left\{ B_{1,}\ldots,B_{d^{2}}\right\} $ satisfying
\begin{enumerate}
\item $B_{j}=c\ket{e_{j}}\bra{e_{j}}$ for $1\leq j\leq d$, with $0<c<1$
and an orthonormal basis $\left\{ \ket{e_{1}},\ldots\ket{e_{d}}\right\} $
of $\C^{d}$;
\item $\sum_{j=1}^{d^{2}}B_{j}\in\mathcal{E}_{d}\,$.
\end{enumerate}
\end{defn}
Given any orthonormal basis $\left\{ \ket{e_{1}},\ldots,\ket{e_{d}}\right\} $
of $\C^{d}$, we can construct an augmented basis for the space of
operators acting on it. First, complete the $d$ projectors 
\begin{equation}
\Pi_{j}=\ket{e_{j}}\bra{e_{j}}\,,\qquad j=1\ldots d\,,\label{eq:cfs1}
\end{equation}
into a basis $\left\{ \Pi_{1},\ldots,\Pi_{d^{2}}\right\} $ of the
Hermitian operators on $\C^{d}$, by adding $d(d-1)$ further rank-one
projections; this is always possible \cite{caves2002unknown}. The
sum 
\begin{equation}
G=\sum_{j=1}^{d^{2}}\Pi_{j}\,,
\end{equation}
is necessarily a positive operator. The relation $\Tr G=d^{2}$ implies
that G must have at least one eigenvalue larger than $1$. If $\Gamma>1$
is the largest eigenvalue of $G$, then $G/\Gamma$ is an effect since
it is a positive operator with eigenvalues less than or equal to one.
Defining 
\begin{equation}
B_{j}=\Pi_{j}/\Gamma\,,\qquad j=1\ldots d^{2}\,,
\end{equation}
the set $\mathcal{B}=\left\{ B_{1},\ldots,B_{d^{2}}\right\} $ turns
into an augmented basis. One can show that $\mathcal{B}$ can never
correspond to a POM. Nevertheless, the effects $B_{j}$ \emph{coexist}
in the sense that they can occur in one single POM, for example $\left\llbracket B_{1},\ldots,B_{d^{2}},\I-G/\Gamma\right\rrbracket $.

Given an effect, one can always represent it as a positive linear
combination of elements in a suitable augmented basis.
\begin{lem}
\label{lem:nonneg decomp}For any effect $E\in\mathcal{E}_{d}$ there
exists an augmented basis $\mathcal{B}$ such that $E$ is in the
positive cone of $\mathcal{B}$.
\end{lem}
\begin{proof}
By the spectral theorem we may write
\begin{equation}
E=\sum_{j=1}^{d}\lambda_{j}\ket{e_{j}}\bra{e_{j}}\,,\qquad\lambda_{j}\in\left[0,1\right]\,,
\end{equation}
for an orthonormal basis $\left\{ \ket{e_{j}},1\leq j\leq d\right\} $
of $\C^{d}$. Take $\mathcal{B}$ to be an augmented basis with 
\begin{equation}
B_{j}=c\ket{e_{j}}\bra{e_{j}}\,,
\end{equation}
for $1\leq j\leq d$ and some $c\in\left(0,1\right)$. Then we may
express $E$ as the linear combination 
\begin{equation}
E=\sum_{j=1}^{d^{2}}e_{j}B_{j}\,,\label{eq: decomp of effect in cone C(B)}
\end{equation}
with non-negative coefficients
\begin{equation}
e_{j}=\begin{cases}
\frac{1}{c}\lambda_{j} & j=1\ldots d\,,\\
0 & j=(d+1)\ldots d^{2}\,,
\end{cases}
\end{equation}
showing that the positive cone of the basis ${\cal B}$ indeed contains
the effect $E$.
\end{proof}
Finally, we need to establish that the intersection of the positive
cones associated with an augmented basis and a MIC-POM, respectively,
has dimension $d^{2}$.
\begin{figure}
\centering{}\includegraphics[scale=0.9]{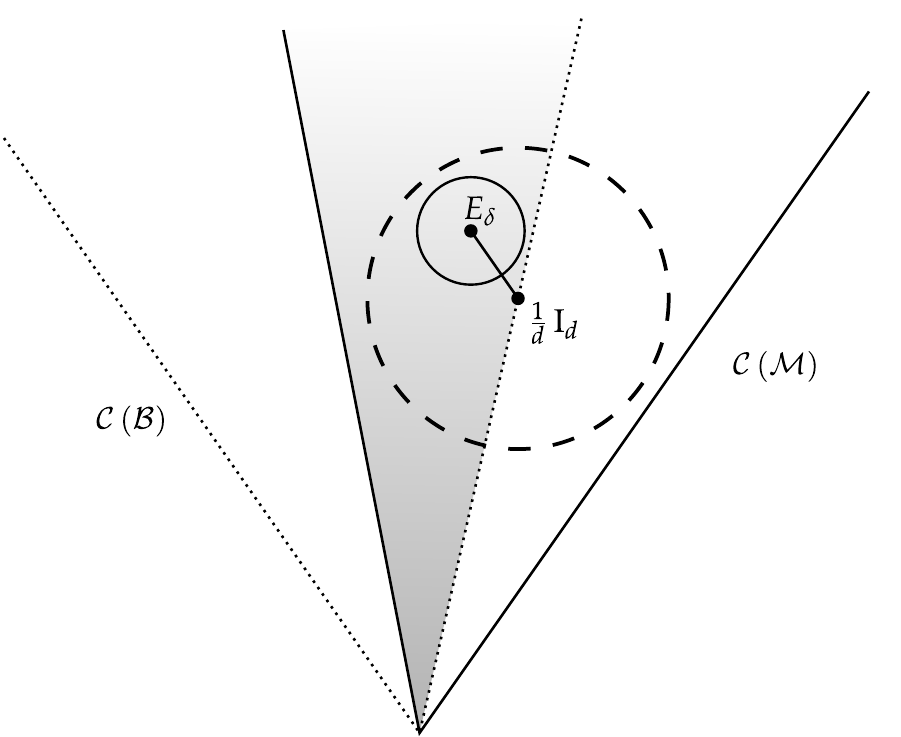}\caption{Sketch of the construction of the open ball $\mathfrak{B}_{\gamma}\left(E_{\delta}\right)$
of dimension $d^{2}$; the positive cones $\mathcal{C}\left(\mathcal{M}\right)$
(solid border) and $\mathcal{C}\left({\cal B}\right)$ (dotted border)
intersect in the cone ${\cal C}({\cal B})\cap{\cal C}({\cal M})$
(shaded cone); the intersection entirely contains the $d^{2}$-dimensional
ball $\mathfrak{B}_{\gamma}\left(E_{\delta}\right)$ around $E_{\delta}$
(solid circle) sitting inside the ball $\mathfrak{B}_{\varepsilon}\left(\I_{d}/d\right)$
of radius $\varepsilon$ around $\I_{d}/d$ (dashed circle); the distance
between $E_{\delta}$ and $\I_{d}/d$ (solid line) is given in Eq.
(\ref{eq: distance between E_delta and I/d}) .\label{fig:cones}}
\end{figure}

\begin{lem}
\label{Lemma cone intersection span}Let $\mathcal{B}=\left\{ B_{1},\ldots,B_{d^{2}}\right\} $
be an augmented basis and $\mathcal{M}=\left\llbracket M_{1},\ldots,M_{d^{2}}\right\rrbracket $
a MIC-POM on $\C^{d}$. The effects in the intersection ${\cal C}({\cal B})\cap{\cal C}({\cal M})$
of the positive cones of $\mathcal{B}$ and $\mathcal{M}$ span the
real vector space $\mathbb{H}_{d}$ of Hermitian operators on $\C^{d}$.
\end{lem}
\begin{proof}
Since the effects in a POM sum to the identity, we have
\begin{equation}
\frac{1}{d}\I_{d}=\sum_{j=1}^{d^{2}}\frac{1}{d}M_{j}\,.\label{eq:uniqueidentity}
\end{equation}
With each of the coefficients in the unique decomposition on the right-hand
side being finite and positive (as opposed to non-negative), the effect
$\I_{d}/d$ is seen to be an \emph{interior} point of the positive
cone $\mathcal{C}\left(\mathcal{M}\right)$. At the same time, the
effect $\I_{d}/d$ is located on the \emph{boundary} of the cone ${\cal C}({\cal B})$
since its expansion in an augmented basis has only $d$ non-zero terms.
Let us define the operator 
\begin{equation}
E_{\delta}=\frac{1}{d}\I_{d}+\delta\sum_{j=d+1}^{d^{2}}B_{j}=\frac{1}{cd}\sum_{j=1}^{d}B_{j}+\delta\sum_{j=d+1}^{d^{2}}B_{j}\,,
\end{equation}
which, for any positive $\delta>0$, is an \emph{interior} point of
the cone $\mathcal{C}\left(\mathcal{B}\right)$: each of the positive
coefficients in its unique decomposition in terms of the augmented
basis ${\cal B}$ is non-zero; we have used Property 1 of Def. \ref{def:An-augmented-basis}
to express the identity $\I_{d}$ in terms of the basis ${\cal B}$.
For sufficiently small values of $\delta$, the operator $E_{\delta}$
is also an interior point of the open ball $\mathfrak{B}_{\varepsilon}\left(\I_{d}/d\right)$
with radius $\varepsilon$ about the point $I_{d}/d$ since
\begin{equation}
\left\Vert E_{\delta}-\frac{1}{d}\I_{d}\right\Vert =\delta\left\Vert \sum_{j=d+1}^{d^{2}}B_{j}\right\Vert <\varepsilon\label{eq: distance between E_delta and I/d}
\end{equation}
holds whenever 
\begin{equation}
0<\delta<\varepsilon\left\Vert \sum_{j=d+1}^{d^{2}}B_{j}\right\Vert ^{-1}\:.
\end{equation}
Being an interior point of both the positive cones $\mathcal{C}\left(\mathcal{B}\right)$
and $\mathcal{C}\left(\mathcal{M}\right)$, the operator $E_{\delta}$
is at the center of an open ball $\mathfrak{B}_{\gamma}\left(E_{\delta}\right)$,
located entirely in the intersection ${\cal C}({\cal B})\cap{\cal C}({\cal M})$
(cf. Fig. \ref{fig:cones}). Since the ball $\mathfrak{B}_{\gamma}\left(E_{\delta}\right)$
has dimension $d^{2},$ the effects contained in it must indeed span
the real vector space $\mathbb{H}_{d}$ of Hermitian operators.
\end{proof}
Combining Theorem \ref{Thm bounded Cauchy} with Lemmata \ref{lem:nonneg decomp}
and \ref{Lemma cone intersection span} will allow us to present a
new proof of Busch's Gleason-type theorem.

\subsection{An alternative proof of Busch's Gleason-type theorem\label{subsec:Alternative-proof}}

Recalling that the trace of the product of two Hermitian operators
constitutes an inner product on the vector space of Hermitian operators,
Theorem \ref{thm: Busch} essentially states that the frame function
$f$ acting on an effect can be written as the inner product of that
effect with a fixed density operator. To underline the connection
with the inner product we adopt the following notation. Let $\mathcal{A}=\left\{ A_{1},\ldots,A_{d^{2}}\right\} $
be a basis for the Hermitian operators $\mathbb{H}_{d}$ on $\C^{d}$.
We describe the effect $E$ by the ``effect vector'' $\mathbf{e}=\left(e_{1},\ldots,e_{d^{2}}\right)^{T}\in\R^{d^{2}}$,
given by its expansion coefficients in this basis,
\begin{equation}
E=\sum_{j=1}^{d^{2}}e_{j}A_{j}\equiv\mathbf{e}\cdot\mathbf{A}\,,\label{eq: expansion of effects in A}
\end{equation}
where $\mathbf{A}$ is an operator-valued vector with $d^{2}$ components.
Theorem \ref{thm: Busch} now states that the frame function is given
by a scalar product, 
\begin{equation}
f\left(E\right)=\mathbf{e}\cdot\mathbf{c}\,,\label{eq:finnerproduct}
\end{equation}
between the effect vector $\mathbf{e}$ and a \emph{fixed} vector
$\mathbf{c}\in\R^{d^{2}}$. Let us determine the relation between
the density matrix $\rho$ in (\ref{eq:busch density-2-2-1}) in the
theorem and the vector $\mathbf{c}$ in (\ref{eq:finnerproduct}).
Consider any orthonormal basis $\mathcal{W}=\left\{ W_{1},\dots,W_{d^{2}}\right\} $
of the Hermitian operators on $\C^{d}$ and let $\mathbf{e}^{\prime}\in\R^{d^{2}}$
be the vector such that $E=\mathbf{e^{\prime}}\cdot\mathbf{W}$. Then
we may write
\begin{align}
f\left(E\right) & =\mathbf{e}\cdot\mathbf{c}=\mathbf{e}^{\prime}\cdot\mathbf{c}^{\prime}=\Tr\left(\sum_{j=1}^{d^{2}}e_{j}^{\prime}W_{j}\sum_{k=1}^{d^{2}}c_{k}^{\prime}W_{k}\right)\nonumber \\
 & =\Tr\left(E\sum_{j=1}^{d^{2}}c_{j}^{\prime}W_{j}\right)\,;\label{eq: expansion of f(E)}
\end{align}
here $\mathbf{c^{\prime}}\in\R^{d^{2}}$ is a fixed vector given by
$\mathbf{c}^{\prime}=C^{-T}\mathbf{c}$ and $C^{-T}$ is the inverse
transpose of the change-of-basis matrix $C$ between the bases $\mathcal{B}$
and $\mathcal{W}$, i.e. the matrix satisfying $C\mathbf{h}=\mathbf{h}^{\prime}$
for all Hermitian operators $H=\mathbf{h}\cdot\mathbf{B}=\mathbf{h^{\prime}}\cdot\mathbf{W}$.
By the definition of a frame function the operator 
\begin{equation}
\rho\equiv\sum_{j=1}^{d^{2}}c_{j}^{\prime}W_{j}=\sum_{j=1}^{d^{2}}\left(C^{-T}\right)_{jk}c_{k}W_{j}
\end{equation}
must be positive semi-definite (since $f$ is positive) and have unit
trace (due to Eq. (\ref{eq:identity})) i.e. be a density operator.

We will now prove that a frame function always admits an expression
as in Eq. (\ref{eq:finnerproduct}).
\begin{proof}
By Lemma \ref{lem:nonneg decomp}, there exists an augmented basis
$\mathcal{B}=\left\{ B_{1},\ldots,B_{d^{2}}\right\} $ for any $E\in\mathcal{E}_{d}$
such that 
\begin{equation}
E=\mathbf{e}\cdot\mathbf{B}\equiv\sum_{j=1}^{d^{2}}e_{j}B_{j}\,,\label{eq: expansion of effects in B}
\end{equation}
with coefficients $e_{j}\geq0$, as in Eq. (\ref{eq: expansion of effects in A}).

For each value $j\in\left\{ 1,\dots,d^{2}\right\} $, we write the
restriction of the frame function $f$ to the set of effects of the
form $xB_{j}$, for $x\in\mathbb{R}$, as
\begin{equation}
f\left(xB_{j}\right)=F_{j}\left(x\right)\,,\label{eq:def Fj}
\end{equation}
where $F_{j}:\left[0,a_{j}\right]\rightarrow\left[0,1\right]$ and
$a_{j}=\max\left\{ x|xB_{j}\in\mathcal{E}_{d}\right\} $. By Eq. (\ref{eq:additivity})
we have that $F_{j}$ satisfies Cauchy's functional equation, i.e.
$F_{j}\left(x+y\right)=F_{j}\left(x\right)+F_{j}\left(y\right)$.
Due to the assumption in Theorem \ref{thm: Busch} that $f:\mathcal{E}_{d}\rightarrow\left[0,1\right]$,
each $F_{j}$ must satisfy Condition \ref{enu: Cauchy 1} of Theorem
\ref{Thm bounded Cauchy} which implies
\begin{equation}
f\left(xB_{j}\right)=F_{j}\left(x\right)=F_{j}\left(1\right)x=f\left(B_{j}\right)x\,.\label{eq:Fj is linear}
\end{equation}
Thus we find
\begin{equation}
f\left(E\right)=\sum_{j=1}^{d^{2}}f\left(e_{j}B_{j}\right)=\sum_{j=1}^{d^{2}}e_{j}f\left(B_{j}\right)=\mathbf{e}\cdot\mathbf{f}_{\mathcal{B}}\,,\label{eq:E in basis B}
\end{equation}
where the $j$-th component of $\mathbf{f}_{\mathcal{B}}\in\R^{d^{2}}$
is given by $f\left(B_{j}\right)$, by repeatedly using additivity
and Eq. (\ref{eq:Fj is linear}). Note that Eq. (\ref{eq:E in basis B})
is not yet in the desired form of Eq. (\ref{eq:finnerproduct}) since
the vector $\mathbf{f}_{\mathcal{B}}$ depends on the basis $\mathcal{B}$
and thus the effect $E$.

Let $\mathcal{M}=\left\llbracket M_{1},\ldots,M_{d^{2}}\right\rrbracket $
be a MIC-POM on $\C^{d}$. Since the elements of $\mathcal{M}$ are
a basis for the space $\mathbb{H}_{d}$, the Hermitian operators on
$\C^{d}$, we have for any $E\in\mathcal{E}_{d}$ 
\begin{equation}
E=\mathbf{e^{\prime\prime}}\cdot\mathbf{M}\,,
\end{equation}
for coefficients $e_{j}^{\prime\prime}\in\mathbb{R}$ some of which
may be negative. There exists a fixed change-of-basis matrix $D$
such that 
\begin{equation}
D\mathbf{e}=\mathbf{e^{\prime\prime}}\,,
\end{equation}
for all effects $E\in\mathcal{E}_{d}$. Now we have 
\begin{equation}
\begin{aligned}f\left(E\right) & =\mathbf{e}\cdot\mathbf{f}_{\mathcal{B}}\\
 & =\left(D\mathbf{e}\right)\cdot\left(D^{-T}\mathbf{f}_{\mathcal{B}}\right)\\
 & =\mathbf{e^{\prime\prime}}\cdot\left(D^{-T}\mathbf{f}_{\mathcal{B}}\right).
\end{aligned}
\label{eq:arbitrary effect fixed basis}
\end{equation}
Any effect $G$ in the intersection of the positive cones $\mathcal{C}\left(\mathcal{B}\right)$
and $\mathcal{C}\left(\mathcal{M}\right)$ can be expressed in two
ways, 
\begin{equation}
G=\mathbf{g}\cdot\mathbf{B}=\mathbf{g}^{\prime\prime}\cdot\mathbf{M}\,,
\end{equation}
where both effect vectors $\mathbf{g}$ and $\mathbf{g}^{\prime\prime}$
have only non-negative components. Eqs. (\ref{eq:E in basis B}) and
(\ref{eq:arbitrary effect fixed basis}) imply that 
\begin{equation}
\begin{aligned}\mathbf{g^{\prime\prime}}\cdot\mathbf{f}_{\mathcal{M}} & =f\left(G\right)=\mathbf{g}^{\prime\prime}\cdot\left(D^{-T}\mathbf{f}_{\mathcal{B}}\right)\,.\end{aligned}
\end{equation}
Since by Lemma \ref{Lemma cone intersection span} there are $d^{2}$
linearly independent effects $G$ in the intersection $\mathcal{C}\left(\mathcal{M}\right)\cap\mathcal{C}\left(\mathcal{B}\right)$,
we conclude that
\begin{equation}
D^{-T}\mathbf{f}_{\mathcal{B}}=\mathbf{f}_{\mathcal{M}}\,.
\end{equation}
Combining this equality with Equation (\ref{eq:arbitrary effect fixed basis})
we find, for a fixed MIC-POM $\mathcal{M}=\left\llbracket M_{1},\ldots,M_{d^{2}}\right\rrbracket $
and any effect $E\in\mathcal{\mathcal{E}}_{\C^{d}}$, that the frame
function $f$ takes the form
\begin{equation}
f\left(E\right)=\mathbf{e}^{\prime\prime}\cdot\mathbf{f}_{\mathcal{M}}\,.
\end{equation}
Here $\mathbf{f}_{\mathcal{M}}\equiv\mathbf{c}$ is a \emph{fixed}
vector since it does not depend on $E$.
\end{proof}
Note that Eq. (\ref{eq:Fj is linear}) may also be found using the
other three cases of Theorem \ref{Thm bounded Cauchy}. For Case \ref{enu:Cauchy 2},
we observe that each of the functions $F_{j},j=1\ldots d^{2}$, is
\emph{non-negative} by definition. Alternatively, each function $F_{j}$
can be shown to be \emph{continuous} at zero (Case \ref{enu:Cauchy 3})
using the following argument which is similar to the one given in
\cite{caves2002unknown}. Assume $F_{j}$ is not continuous at zero.
Then there exists a number $\varepsilon>0$ such that for all $\delta>0$
we have
\begin{equation}
F_{j}\left(x_{0}\right)>\varepsilon\,,
\end{equation}
for some $0<x_{0}<\delta<1$. For any given $\varepsilon$ choose
$\delta=1/n<\varepsilon$, there is a value of $x_{0}<\delta$ such
that $F_{j}\left(x_{0}\right)>\varepsilon$. However, we have the
inequality $nx_{0}<1$, which leads to
\begin{equation}
F_{j}\left(nx_{0}\right)=nF_{j}\left(x_{0}\right)>n\varepsilon>1\,,
\end{equation}
 contradicting the the existence of an upper bound of one on values
of $F_{j}$. Finally, each of the functions $F_{j}$ is \emph{Lebesgue
measurable} (Case \ref{enu:Cauchy 4}) which follows from the monotonicity
of the function.

\section{Summary and discussion\label{sec:Summary-and-discussion}}

We are aware of two papers linking Gleason's theorem and Cauchy's
functional equation. Cooke et al. \cite{cooke_keane_moran_1985} used
Cauchy's functional equation to demonstrate the necessity of the boundedness
of frame functions in proving Gleason's theorem. Dvu\-re\-\v{c}enskij
\cite{Dvurecenskij1996} introduced frame functions defined on effect
algebras but did not proceed to derive a Gleason-type theorem in the
context of quantum theory.

In this paper, we have exploited the fact that additive functions
are central to both Gleason-type theorems and Cauchy's functional
equation. Gleason-type theorems are based on the assumption that states
assign probabilities to measurement outcomes via additive functions,
or \emph{frame functions,} on the effect space. Linearity of the frame
functions\emph{ }has been shown to follow from positivity and other
assumptions which are well-known in the context of Cauchy's functional
equation. Altogether, the result obtained here amounts to an alternative
proof of the extension of Gleason's theorem to dimension two given
by Busch \cite{Busch2003} and Caves et al. \cite{caves2002unknown}.

Other\emph{ }Gleason-type theorems are known that are \emph{stronger,
}in the sense that they depend on assumptions \emph{weaker} than those
of Theorem \ref{thm: Busch}. The smallest known set of assumptions
requires Eq. (\ref{eq:additivity}) to only be valid for effects $E_{1}$
and $E_{2}$ that coexist in a \emph{projective-simulable} POM \cite{Oszmaniec2017},
i.e. a POM that may be simulated using only classic mixtures of projective
measurements, as opposed to any POM. Since the proof given in \cite{Wright2018}
relies on Theorem \ref{thm: Busch}, the alternative proof presented
in Sec. \ref{subsec:Alternative-proof} also gives rise to a new proof
of the strongest existing Gleason-type theorem.

We have not been able to exploit the structural similarity between
the requirements on frame functions and on the solutions of Cauchy's
functional equation in order to yield a new proof of Gleason's original
theorem. Additivity of frame functions defined on projections instead
of effects does not provide us with the type of continuous parameters
that are necessary for the argument developed here. It remains an
intriguing open question whether such a proof does exist.
\begin{acknowledgement*}
\end{acknowledgement*}
The authors thank Jonathan Barrett for pointing out a gap in the proof
of Theorem 1 given in an earlier version of this paper. VJW gratefully
acknowledges funding from the York Centre for Quantum Technologies
and the WW Smith fund.

\appendix

\section{Proofs of Cases (ii), (iii) and (iv) of Theorem \ref{Thm bounded Cauchy}\label{sec:Appendix}}

It is shown that each of the conditions given in Cases (ii) to (iv)
imply Theorem \ref{Thm bounded Cauchy} which states that an additive
function on a particular interval must be linear.
\begin{proof}
Case \ref{enu:Cauchy 2}: Suppose that there exists a \emph{non-linear}
function $f$ satisfying Eq. (\ref{eq:additivethm}) and Case \ref{enu:Cauchy 2}
of Theorem \ref{Thm bounded Cauchy}. Then the function $g:\left[0,a\right]\rightarrow\mathbb{R}$
defined by $g\left(x\right)=-f\left(x\right)$ is non-linear but satisfies
Eq. (\ref{eq:additivethm}) and $g\left(x\right)\leq b$ and $b\geq0$,
with $b=-c$, contradicting Case \ref{enu: Cauchy 1}.
\end{proof}
\begin{proof}
Case \ref{enu:Cauchy 3}: Since $f$ is continuous at zero and $f\left(0\right)=0$,
as follows from Eq. (\ref{eq:additivethm}), we have that for any
$\varepsilon>0$, there exists a $\delta>0$ such that $\left|f\left(x\right)\right|<\varepsilon$
for all $x$ satisfying $\left|x\right|<\delta$. Let $x,x_{0}\in\left[0,a\right]$
be such that $\left|x-x_{0}\right|<\delta$. First consider the case
$x<x_{0}$. Using additivity, 
\begin{equation}
f\left(x\right)+f\left(x_{0}-x\right)=f\left(x+x_{0}-x\right)=f\left(x_{0}\right)\,,\label{eq:-vecont}
\end{equation}
we find 
\begin{equation}
\left|f\left(x\right)-f\left(x_{0}\right)\right|=\left|f\left(x_{0}-x\right)\right|<\varepsilon\,.
\end{equation}
On the other hand, if $x>x_{0}$ we have 
\begin{equation}
f\left(x\right)=f\left(x-x_{0}+x_{0}\right)=f\left(x-x_{0}\right)+f\left(x_{0}\right)\,,
\end{equation}
and then 
\begin{equation}
\left|f\left(x\right)-f\left(x_{0}\right)\right|=\left|f\left(x-x_{0}\right)\right|<\varepsilon\,.
\end{equation}
It follows that $f$ is continuous on $\left[0,a\right]$. As in the
proof for Case \ref{enu: Cauchy 1}, Eqs. (\ref{eq:rational1}) and
(\ref{eq:rational2}) show that 
\begin{equation}
f\left(q\right)=f\left(1\right)q\,,
\end{equation}
for rational $q\in\left[0,a\right]$. Therefore, if $\left(q_{1},q_{2},\ldots\right)$
is a sequence of rational numbers converging to $x$, the function
$f(x)$ must be linear in $x$: 
\begin{equation}
f\left(x\right)=\lim_{j\rightarrow\infty}f\left(q_{j}\right)=\lim_{j\rightarrow\infty}f\left(1\right)q_{j}=f\left(1\right)x\,.
\end{equation}
\end{proof}
In Case \ref{enu:Cauchy 4}, where $f$ is Lebesgue measurable, the
proof of the analogous result for functions on the full real line
by Banach \cite{banach1920equation} is easily adapted to our setting.
Given Case \ref{enu:Cauchy 3}, it suffices to prove that $f$ is
continuous at $0$, i.e. that for every $\varepsilon>0$ there exists
a number $\delta>0$ such that 
\begin{equation}
\left|f\left(h\right)-f\left(0\right)\right|=\left|f\left(h\right)\right|<\varepsilon
\end{equation}
holds for all $0<h<\delta$.
\begin{proof}
Case \ref{enu:Cauchy 4}: Let $a/2<r<a$. Lusin's theorem \cite{lusin1912proprietes}
states that, for a Lebesgue measurable function $g$ on an interval
$J$ of Lesbesgue measure $\mu\left(J\right)=m$, there exists a compact
subset of any measure $m^{\prime}<m$ such that the restriction of
$g$ to this subset is continuous. Thus we may find a compact set
$F\subset\left[0,a\right]$ with $\mu\left(F\right)\geq r$ on which
$f$ is continuous. Let $\varepsilon>0$ be given. Since $F$ is compact,
$f$ is uniformly continuous on $F$ and there exists a $\delta\in\left(0,2r-a\right)$
such that 
\begin{equation}
\left|f\left(x\right)-f\left(y\right)\right|<\varepsilon
\end{equation}
is valid for two numbers $x,y\in F$ such that $\left|x-y\right|<\delta$.
Let $h\in\left(0,\delta\right)$. Suppose $F$ and $F-h=\left\{ x-h|x\in F\right\} $
were disjoint. Then we would have 
\begin{equation}
a+h=\mu\left(\left[-h,a\right]\right)\geq\mu\left(F\cup\left(F-h\right)\right)=\mu\left(F\right)+\mu\left(F-h\right)\geq2r\,,
\end{equation}
which contradicts $h<\delta<2r-a$. Taking a point $x\in F\cap\left(F-h\right)$
then a number $\delta\in\left(0,2r-a\right)$ can be found such that
\begin{equation}
\begin{aligned}\left|f\left(h\right)\right|=\left|f\left(x\right)-f\left(x\right)-f\left(h\right)\right|=\left|f\left(x\right)-f\left(x+h\right)\right| & <\varepsilon\,,\end{aligned}
\end{equation}
for $h\in\left(0,\delta\right)$. Hence, remembering that $f(0)=0$,
the function $f(x)$ is continuous at $x=0$.
\end{proof}

\end{document}